\documentclass[letterpaper, 10 pt, journal, twoside]{IEEEtran}

\usepackage{cite}
\usepackage{amsmath,amssymb,amsfonts,amsthm}
\usepackage{algorithmic}
\usepackage{graphicx}
\usepackage{textcomp}
\usepackage{xcolor}
\usepackage{algorithm}
\usepackage{algorithmic}

\newcommand{\g}{\mathfrak{g}} 

\newtheorem{theorem}{Theorem}
\newtheorem{lemma}{Lemma}

\newtheorem{corollary}{Corollary}
\newtheorem{definition}{Definition}
\newtheorem{remark}{Remark}

\title{Space-Time Diversity in Observability and Estimation on Product Lie Groups}

\author{Somasundhar~Venkatasubramanian,  Anirudh~Venkat and Advaidh~Venkat
\thanks{Dr. S. Venkatasubramanian is the Chief of Research at Gramian Labs, Ann Arbor, Michigan, USA (e-mail: soma@gramianlabs.com).}
\thanks{A. Venkat is a Member of Technical Staff at Gramian Labs (e-mail: anirudh@gramianlabs.com).}
\thanks{A. Venkat is a Vice President at Gramian Labs (e-mail: advaidh@gramianlabs.com).}}

\begin{document}

\maketitle

\begin{abstract}
Robust state estimation in coupled dynamical systems depends critically not only on sensor quality but on the structural alignment between observation channels and the system’s intrinsic dynamics. This paper develops a rigorous framework for analyzing spatial and temporal diversity in dynamical state estimation on product Lie groups, drawing structural parallels to diversity gains in space-time coding. Three main results are established: (i) coupling-based necessary and sufficient conditions for cross-factor observability, showing that a sensor local to one group factor renders another factor observable if and only if the dynamics propagate error directions across the corresponding Lie algebra components; (ii) a spatial diversity saturation theorem identifying precisely when additional observation channels fail to expand the propagated observation subspace and thus provide no structural benefit; and (iii) a time-space diversity decomposition that exactly separates instantaneous spatial information from accumulated temporal information in the estimation error covariance. The framework is applied to planar SE(2) and spatial SE(3) navigation, yielding exact observability guarantees for redundant and non-redundant sensor architectures in modern robotics and autonomous vehicles. These results extend classical observability theory beyond Euclidean state spaces, exposing structural constraints invisible to standard rank-based analysis that fundamentally govern robust inference in coupled dynamical systems.
\end{abstract}


\section{Introduction}
\IEEEPARstart{R}{obust} inference under noise and partial observability is a central problem in both communication theory and dynamical systems. In wireless communication, space-time coding mitigates fading by distributing information across multiple transmit antennas and time slots, yielding diversity gains governed by rank and eigenvalue properties of the induced channel matrices rather than by modulation details. Space-time coding and diversity principles were developed in the context of fading channels, where robustness is achieved by distributing information across antennas and time slots \cite{alamouti1998simple, tarokh1999spacetime}.

In dynamical state estimation, robustness is traditionally analyzed through observability and filtering theory. Classical results establish that estimation accuracy depends on the conditioning of the observability Gramian, which aggregates information across time \cite{kailath1980linear, anderson1979optimal}. However, two structural questions remain largely unresolved:
\begin{enumerate}
    \item When does spatial diversity across heterogeneous observation channels genuinely increase observability in coupled dynamical systems?
    \item When does the addition of further sensors cease to provide additional benefit?
\end{enumerate}

These questions are particularly salient for systems evolving on product Lie groups, where distinct components of the state---such as orientation and position---may be observed independently while remaining coupled through the dynamics. In such systems, sensing alone cannot create information flow across state components; any cross-factor observability must arise from the structure of the dynamics themselves. For example, in a planar motion system, position-only sensing can render orientation observable over time through motion-induced coupling, whereas a decoupled component would remain unobservable regardless of sensing duration.

It is crucial to distinguish this structural problem from the mechanics of the filtering algorithms themselves. Recent breakthroughs in estimation on Lie groups, notably the Invariant Extended Kalman Filter (IEKF) \cite{barrau2017invariant}, have established that exploiting symmetries in the system dynamics yields state-independent error propagation and remarkably stable observer performance. However, the IEKF literature primarily addresses \textit{how} to optimally fuse available measurements given an assumed sensor suite. It does not answer the \textit{a priori} architectural question: which combinations of sensors are structurally capable of rendering the coupled state space observable? While the IEKF guarantees that the error dynamics possess a log-linear structure, our framework provides the necessary and sufficient structural conditions that determine whether the corresponding observability Gramian is non-degenerate. In essence, while invariant filtering provides the computational engine for robust inference, the present space-time diversity theory dictates the fundamental informational limits of the sensor architecture.

Our work relates broadly to nonlinear observability theory \cite{hermann1977nonlinear} and invariant filtering \cite{barrau2017invariant}. While visual-inertial SLAM literature extensively analyzes observability using deterministic Lie derivatives \cite{martinelli2012observability}, these analyses typically focus on local distinguishability rather than the structural saturation limits of adding redundant sensors. Similarly, while the Posterior Cram\'{e}r-Rao Lower Bound (PCRLB) quantifies fundamental performance limits for nonlinear filtering \cite{tichavsky1998posterior}, our framework establishes the underlying structural conditions—specifically, whether the Fisher information strictly increases or structurally saturates—dictated by the Lie group's topology. Our contribution is a structural analysis of observability on product Lie groups, where the key issue is not merely how measurements are fused, but whether the dynamics themselves propagate information across group factors and whether additional sensors expand the observable subspace.

In this setting, factor-local sensing can reveal other state components only through dynamical coupling in the system evolution.

First, we derive necessary and sufficient conditions for cross-factor observability under factor-local sensing \textbf{(Section III, Theorem 1)}, showing that observability of one group component by measurements on another is possible if and only if the dynamics propagate error directions across the corresponding Lie algebra factors. Second, we establish a spatial diversity saturation theorem \textbf{(Section IV, Theorem 2)}, formalizing diminishing returns in sensor diversity by characterizing when additional observation channels fail to expand the propagated observation subspace. Third, we derive a time-space diversity decomposition \textbf{(Section V, Theorem 3)} that separates uncertainty reduction into instantaneous spatial information and accumulated temporal information.

Together, these results provide a structural theory of diversity in dynamical estimation that goes beyond monotonicity arguments and reinterpretations of classical filtering results. Unlike classical observability analyses for interconnected or multi-sensor systems, the present results provide necessary-and-sufficient structural conditions for information propagation and saturation that arise specifically from the product-Lie-group decomposition and cannot be obtained by rank-based Gramian reinterpretation alone. Specifically, the structural conditions explicitly map observability to the off-diagonal transition blocks $\Phi_{j\leftarrow i}$, revealing the causal dependencies required for robust inference that are invisible to a purely numerical rank check.

\begin{table}[htbp]
\caption{Comparison of Observability Frameworks}
\label{tab:comparison}
\centering
\renewcommand{\arraystretch}{1.4} 
\begin{tabular}{p{0.2\columnwidth} p{0.35\columnwidth} p{0.35\columnwidth}}
\hline
\textbf{Framework} & \textbf{Core Mechanism} & \textbf{Diversity Limit} \\
\hline
Classical \cite{kailath1980linear} & Euclidean rank conditions & Bounded by state dim. $n$ \\
IEKF \cite{barrau2017invariant} & Log-linear error propagation & Assumes given sensors \\
SLAM \cite{martinelli2012observability} & Deterministic Lie derivatives & Focuses on local rank \\
\textbf{Proposed} & \textbf{Cross-factor transition blocks} & \textbf{Subspace saturation ($\mathcal{S}_K$)} \\
\hline
\end{tabular}
\end{table}

\section{Mathematical Setting and Problem Formulation}
Let $G = \prod_{i=1}^{m} G_i$ be a finite-dimensional product Lie group with Lie algebra $\g = \bigoplus_{i=1}^{m} \g_i$. We equip $\g$ with an inner product induced by a left-invariant Riemannian metric on $G$.

Consider left-invariant linearized error dynamics of the form:
\begin{equation}
    \delta x_{t+1} = F_t \delta x_t + w_t, \quad \delta x_t \in \g,
\end{equation}
where $w_t \sim \mathcal{N}(0, Q_t)$ is a zero-mean Gaussian process noise sequence with covariance $Q_t$. Assume there are $K$ independent observation channels. The measurements for each channel $k \in \{1, \dots, K\}$ are given by:
\begin{equation}
    \delta y_t^{(k)} = H_t^{(k)} \delta x_t + v_t^{(k)},
\end{equation}
where $v_t^{(k)} \sim \mathcal{N}(0, R_k)$ are mutually independent Gaussian measurement noises with covariance $R_k \succ 0$. Invariant error formulations and their associated linearized dynamics on Lie groups have been extensively studied in the context of observer and filtering design \cite{barrau2017invariant}.

\textbf{Problem Statement:} The fundamental problem addressed in this paper is to characterize how the spatial distribution of the observation channels $H_t^{(k)}$ across distinct subgroup factors $\g_i$ interacts with the dynamical coupling in $F_t$ to determine the overall observability of the system. Specifically, we seek to establish rigorous structural conditions under which spatial diversity yields provable information gain, identify when the addition of further sensors ceases to improve observability, and exactly quantify the resulting reduction in estimation uncertainty.

Let $\Phi(t, \tau)$ denote the state transition operator associated with $\{F_t\}$. The finite-horizon observability Gramian is defined as:
\begin{equation}
    W_o(T) = \sum_{t=0}^{T-1} \Phi(t,0)^\top \left( \sum_{k=1}^{K} H_t^{(k)\top} R_k^{-1} H_t^{(k)} \right) \Phi(t,0), \label{eq:gramian_def}
\end{equation}
where $\Phi(t,0)$ is assumed uniformly bounded.

\section{Coupling Conditions for Cross-Factor Observability}

\subsection{Linearized Invariant Error System}
Let us fix an inner product $\langle \cdot, \cdot \rangle$ on the Lie algebra $\g$. We identify $\g \cong \mathbb{R}^n$ via an orthonormal basis that respects the direct sum decomposition. The state transition operator is defined as $\Phi(t, \tau) = F_{t-1} \cdots F_\tau$ if $t > \tau$, and $I$ if $t=\tau$.

To analyze cross-factor observability, we isolate the Gramian contribution of a single observation channel (dropping the superscript $k$ for brevity). Recall from Section II that the Lie algebra decomposes into $m$ distinct components. For any target factor index $j \in \{1, \dots, m\}$, define the canonical projection $\Pi_j: \g \to \g_j$ and injection $\iota_j: \g_j \to \g$.

\subsection{Block Coupling in the Transition Operator}
Assume the measurement operator is localized to the $i$-th subgroup, meaning it depends exclusively on that specific factor. Mathematically, this is expressed as:
\begin{equation}
    H_t = \tilde{H}_t \Pi_i,
\end{equation}
where $\tilde{H}_t : \g_i \to \mathbb{R}^p$ is the local observation map acting strictly on the $i$-th Lie algebra component.

\begin{definition}[Persistent Excitation and Nondegenerate Sensing]
A factor-local observation channel acting on $G_i$ with linearized operator $H_t = \tilde H_t \Pi_i$ is said to be persistently exciting on $\g_i$ if there exists a finite horizon $T \ge 1$ such that:
\begin{equation}
\sum_{t=0}^{T-1} \Phi_{i\leftarrow i}(t,0)^\top \tilde H_t^\top R_t^{-1}\tilde H_t \Phi_{i\leftarrow i}(t,0) \succ 0.
\end{equation}
Physically, this condition requires that the local sensor collects information along all independent directions of the subgroup's tangent space $\g_i$ over the horizon $T$. We call the sensor nondegenerate on $\g_i$ if it satisfies this condition.
\end{definition}

Define the cross-factor transition block from factor $i$ to factor $j$ as:
\begin{equation}
    \Phi_{j \leftarrow i}(t,0) := \Pi_j \Phi(t,0) \iota_i : \g_i \to \g_j.
\end{equation}

\begin{lemma}[Support of Gramian Blocks]
Let $H_t = \tilde{H}_t \Pi_i$. Then, for any $j \in \{1, \dots, m\}$, the $j$-block of the Gramian satisfies:
\begin{equation}
    \Pi_j W_o(T) \Pi_j^\top = \sum_{t=0}^{T-1} \Phi_{j \leftarrow i}(t,0)^\top \tilde{H}_t^\top R_t^{-1} \tilde{H}_t \Phi_{j \leftarrow i}(t,0).
    \label{eq:gramian_block}
\end{equation}
\end{lemma}
\begin{proof}
Substituting the factor-local measurement $H_t = \tilde{H}_t \Pi_i$ into the single-channel Gramian formulation derived from \eqref{eq:gramian_def}, left- and right-multiplying by $\Pi_j$ and $\Pi_j^\top$, and using the identity $(\Pi_i \Phi(t,0) \Pi_j^\top)^\top = \Pi_j \Phi(t,0)^\top \Pi_i^\top$ yields the result.
\end{proof}

\begin{lemma}[Equivalence of Propagation]
Suppose the factor-local sensor acting on $G_i$ is persistently exciting on $\g_i$. Assume $R_t \succ 0$. Then, for all horizons $T \ge 1$:
\begin{equation}
  \begin{aligned}
    \Pi_j W_o(T)\Pi_j^\top = 0 \quad 
    \text{\textbf{if and only if}} \\
    \Phi_{j\leftarrow i}(t,0) = 0 \quad \quad  \text{for all } t \ge 0
   \end{aligned}
\end{equation}
\end{lemma}
\begin{proof}
$(\Rightarrow)$ Since $R_t^{-1}$ is positive definite, the sum of quadratic forms in \eqref{eq:gramian_block} is zero if and only if each term is zero. Given persistent excitation of $\tilde{H}_t$, this forces $\Phi_{j \leftarrow i}(t,0) = 0$.
$(\Leftarrow)$ If $\Phi_{j \leftarrow i}(t,0) \equiv 0$, the sum is trivially zero.
\end{proof}

\begin{theorem}[Necessary and Sufficient Coupling]
Fix factors $i \neq j$. Assume factor-local sensing on factor $i$ that is nondegenerate on $\g_i$. Then the following are equivalent:
\begin{enumerate}
    \item[(a)] There exists a horizon $T$ such that the projected observability Gramian $\Pi_j W_o(T)\Pi_j^\top$ is positive definite on the subspace:
    \begin{equation}
    \mathcal R_{j\leftarrow i}(T) := \mathrm{span}\{\Phi_{j\leftarrow i}(t,0)u : u \in \mathfrak g_i,\ 0 \le t \le T-1\}.
    \end{equation}
    \item[(b)] There exists $t \ge 0$ such that the cross-factor transition block is nonzero: 
    \begin{equation}
        \Phi_{j \leftarrow i}(t,0) \neq 0.
    \end{equation}
\end{enumerate}
\end{theorem}
\begin{proof}
$(b \Rightarrow a)$ If $\Phi_{j \leftarrow i}(t^*, 0) \neq 0$, the nondegeneracy of sensing on $\g_i$ ensures the quadratic form is strictly positive definite on the range of the transition block.
$(a \Rightarrow b)$ Follows by contradiction from Lemma 2.
\end{proof}
This result shows that cross-factor observability is fundamentally a property of dynamical coupling rather than sensing alone.

\begin{corollary}[Full Cross-Factor Observability]
Under the assumptions of Theorem 1, the factor-local sensor acting on $G_i$ renders the entire factor $\g_j$ observable over horizon $T$ if and only if:
\begin{equation}
\mathcal R_{j\leftarrow i}(T) = \mathfrak g_j.
\end{equation}
\end{corollary}

This theorem formalizes that on a product state space, a sensor local to one factor cannot ``create'' observability of another factor unless the dynamics propagate errors between the corresponding Lie algebra components.

\begin{remark}[Cross-Factor Observability Index]
In classical Euclidean systems, the observability index defines the minimum number of time steps required to achieve full rank in the Gramian. In the product Lie group setting, this index is strictly lower-bounded by the Lie algebraic distance between the measured factor $\g_i$ and the unmeasured factor $\g_j$. If the factors are deeply nested within the system's coupling hierarchy, temporal accumulation must proceed through multiple intermediate transitions before $\mathcal R_{j\leftarrow i}(T)$ spans $\g_j$, significantly delaying the onset of robust estimation.
\end{remark}

\section{Spatial Diversity Saturation}
We now characterize when adding sensors ceases to provide benefit. Let $\{H_t^{(k)}\}_{k=1}^K$ denote $K$ observation channels and let:
\begin{equation}
    \mathcal{S}_K := \overline{\text{span}} \{ \Phi(t,0)^\top H_\ell^\top : \ell=1,\dots,K, t \ge 0 \} \subset \g
\end{equation}
denote the closed linear span of propagated observation subspaces.

\begin{theorem}[Spatial Diversity Saturation]
For any additional channel $H_t^{(K+1)}$,
\begin{equation}
    W_o^{(K+1)}(T) = W_o^{(K)}(T) \quad \forall T
\end{equation}
if and only if:
\begin{equation}
    \text{Range}((H_{K+1})^\top) \subseteq \mathcal{S}_K.
\end{equation}
Intuitively, additional sensors provide benefit only if they introduce new independent directions beyond those already spanned by the propagated observation subspace.
\end{theorem}
\begin{proof}
By definition, the updated observability Gramian incorporating the candidate sensor is:
\begin{equation}
  \begin{split}
      W_o^{(K+1)}(T) = & W_o^{(K)}(T)  \\
      &+ \sum_{t=0}^{T-1} \Phi(t,0)^\top H_{K+1}^\top R_{K+1}^{-1} H_{K+1} \Phi(t,0) \label{eq:expanded_gramian}
 \end{split}
\end{equation}
(Sufficient): Suppose $\text{Range}((H_{K+1})^\top) \subseteq \mathcal{S}_K$. By the definition of $\mathcal{S}_K$, any vector $v$ in the null space of $W_o^{(K)}(T)$ is orthogonal to $\mathcal{S}_K$. Therefore, $H_{K+1} \Phi(t,0) v = 0$ for all $t$. Evaluating the quadratic form $v^\top W_o^{(K+1)}(T) v$ using \eqref{eq:expanded_gramian} yields $0$, meaning the null space of the Gramian is unchanged and no structural information is gained.

(Necessary): Suppose there exists a component of $H_{K+1}^\top$ orthogonal to $\mathcal{S}_K$. Then there exists a state direction $v \in \ker W_o^{(K)}(T)$ such that $H_{K+1} \Phi(t^*,0) v \neq 0$ for some $t^* < T$. Since $R_{K+1} \succ 0$, the sum of quadratic forms strictly increases along the direction $v$, shrinking the unobservable subspace and strictly increasing the rank (or minimizing the minimum eigenvalue) of the observability Gramian.
\end{proof}

Equivalently, the Gramian increment induced by the $(K+1)$-st observation channel satisfies:
\begin{equation}
W_o^{(K+1)}(T) - W_o^{(K)}(T) = 0
\end{equation}
if and only if the added channel contributes no new directions to the propagated observation subspace $\mathcal S_K$. When this condition fails, the rank of the observability Gramian and its nonzero eigenvalues strictly increase for sufficiently large horizons. In the stochastic setting, this corresponds to a strict increase in accumulated Fisher information in the information-form representation.

The saturation conditions in Theorem 2 directly inform optimal sensor selection, bypassing the need for expensive empirical covariance simulations. Algorithm 1 outlines a computationally efficient, geometry-aware procedure for evaluating candidate sensor additions.

\begin{algorithm}[htbp]
\caption{Structural Diversity Check for Sensor Selection}
\label{alg:sensor_selection}
\begin{algorithmic}[1]
\REQUIRE Base transition operator $\Phi(t,0)$ for $0 \le t < T$
\REQUIRE Existing sensor suite $\mathcal{H}_{base} = \{H^{(1)}, \dots, H^{(K)}\}$
\REQUIRE Candidate sensor $H^{(K+1)}$ targeting factor $\g_j$
\STATE \textbf{Initialize:} Subspace $\mathcal{S}_K \leftarrow \{0\}$
\FOR{$k = 1$ \TO $K$}
    \FOR{$t = 0$ \TO $T-1$}
        \STATE $\mathcal{S}_K \leftarrow \mathcal{S}_K \cup \text{span}(\Phi(t,0)^\top H^{(k)\top})$
    \ENDFOR
\ENDFOR
\STATE Orthogonalize basis of $\mathcal{S}_K$
\STATE Compute projection matrix $P_{\mathcal{S}_K}$ onto $\mathcal{S}_K$
\STATE \textbf{Evaluate Candidate:}
\STATE Compute innovation subspace $\mathcal{I} = (I - P_{\mathcal{S}_K}) H^{(K+1)\top}$
\IF{$\|\mathcal{I}\| > \epsilon$}
    \RETURN \TRUE \quad \textit{// Structural diversity gain achieved}
\ELSE
    \RETURN \FALSE \quad \textit{// Sensor is structurally redundant}
\ENDIF
\end{algorithmic}
\end{algorithm}

\section{Time-Space Diversity Decomposition}
Consider the discrete-time invariant Kalman filter for the linearized error system. Let $\delta \hat{x}_t$ denote the state estimate and define the posterior estimation error covariance as: 
\begin{equation}
    P_t := \mathbb{E}[(\delta x_t - \delta \hat{x}_t)(\delta x_t - \delta \hat{x}_t)^\top \mid y_{0:t}].
\end{equation}
The corresponding information matrix is $J_t := P_t^{-1}$.

\begin{theorem}[Time-Space Decomposition]
The one-step change in uncertainty volume satisfies:
\begin{equation}
    \log \det P_{t+1} - \log \det P_t = \log \det (I + Q_t J_t) - \log \det (I + P_t S_t),
\end{equation}
where $S_t := \sum_{k=1}^K H_t^{(k)\top} R_k^{-1} H_t^{(k)}$ is the instantaneous spatial information and $Q_t$ is the process noise covariance.
\end{theorem}
\begin{proof}
Let 
\begin{align}
 \nonumber
   P_{t+1|t} = F_t P_t F_t^\top + Q_t
\end{align}
denote the prior estimation covariance. In the information form, the measurement update is:
\begin{align}
 \nonumber
  J_{t+1} = P_{t+1|t}^{-1} + S_t
\end{align}
  
  Applying the matrix determinant lemma yields 
\begin{equation}
\det(J_{t+1}) = \det(P_{t+1|t}^{-1}) \det(I + P_{t+1|t} S_t)
\end{equation}

Taking the logarithm and negating gives the posterior uncertainty volume:
\begin{equation}
    \log \det P_{t+1} = \log \det P_{t+1|t} - \log \det (I + P_{t+1|t} S_t) \label{eq:proof_step1}
\end{equation}

Furthermore, the prior covariance can be factored as 
\begin{equation}
P_{t+1|t} = (F_t P_t F_t^\top)(I + (F_t P_t F_t^\top)^{-1} Q_t)
\end{equation}

Because $F_t$ is the adjoint representation of a transition operator on a unimodular Lie group, $\det(F_t) = 1$. Consequently, the prior volume expands as:
\begin{equation}
    \log \det P_{t+1|t} = \log \det P_t + \log \det (I + Q_t J_t). \label{eq:proof_step2}
\end{equation}
Substituting \eqref{eq:proof_step2} into \eqref{eq:proof_step1} isolates the temporal and spatial components, completing the proof.
\end{proof}

The first term accounts for temporal dynamics (process noise), while the second term captures the spatial diversity contribution (sensor information). Note that when $G = \mathbb{R}^n$ (Euclidean case), the product decomposition is trivial, the exponential and logarithm maps reduce to the identity, and the invariant error coincides with the linear error. Further, all the results reduce exactly to classical observability and Kalman filtering results. The value of the present framework lies in identifying which aspects of these results persist, and which become structurally constrained, when the state space is non-Euclidean and decomposes into coupled group factors.

\begin{remark}[Limits of Exact Additivity]
While invariant errors possess large regions of log-linear behavior \cite{barrau2017invariant}, the exact additivity in Theorem 3 holds strictly for the linearized system. Under severe nonlinearities, massive initial uncertainty $\Sigma_0$, or exceptionally high noise-to-signal ratios, the higher-order Baker-Campbell-Hausdorff (BCH) terms become non-negligible, degrading this decomposition. However, the structural non-observability guarantees (Theorems 1 and 2) remain fundamental algebraic constraints on the system's information flow regardless of the specific noise regime.
\end{remark}

\section{Application: Diversity on $SE(2)$}
We apply the theory to the kinematics of a planar mobile robot whose pose (heading and 2D position) evolves on the special Euclidean group $SE(2) \cong SO(2) \times \mathbb{R}^2$. The robot is driven by noisy velocity inputs (e.g., forward speed and yaw rate), which dynamically couple the translational and rotational uncertainty. 

The continuous-time kinematics are given by the standard unicycle model:
\begin{align}
    \dot{x} &= v_t \cos \theta_t \\
    \dot{y} &= v_t \sin \theta_t \\
    \dot{\theta} &= \omega_t
\end{align}
These nonlinear dynamics induce a strict hierarchical coupling. Because the translational rates ($\dot{x}, \dot{y}$) depend on the current orientation ($\theta_t$), the corresponding linearized discrete-time transition operator $\Phi(t+1, t)$ develops non-zero off-diagonal blocks $\Phi_{\mathbb{R}^2 \leftarrow \mathfrak{so}(2)}$. This creates the necessary mechanism for cross-factor observability.

We consider three sensing architectures:
\begin{itemize}
    \item \textbf{Scenario A (Time Diversity Only):} The robot is equipped only with a position sensor (e.g., GPS). Heading must be inferred indirectly over time through the dynamical coupling.
    \item \textbf{Scenario B (Space-Time Diversity):} The robot is equipped with both position and a high-precision heading sensor (e.g., GPS and a compass). This provides non-redundant spatial excitation across both group factors.
    \item \textbf{Scenario C (Redundant Spatial Sensing):} A second position sensor is added. Its linearized operator lies in the same propagated subspace as the first position sensor.
\end{itemize}

\begin{figure}[htbp]
    \centering
    \includegraphics[width=0.95\columnwidth]{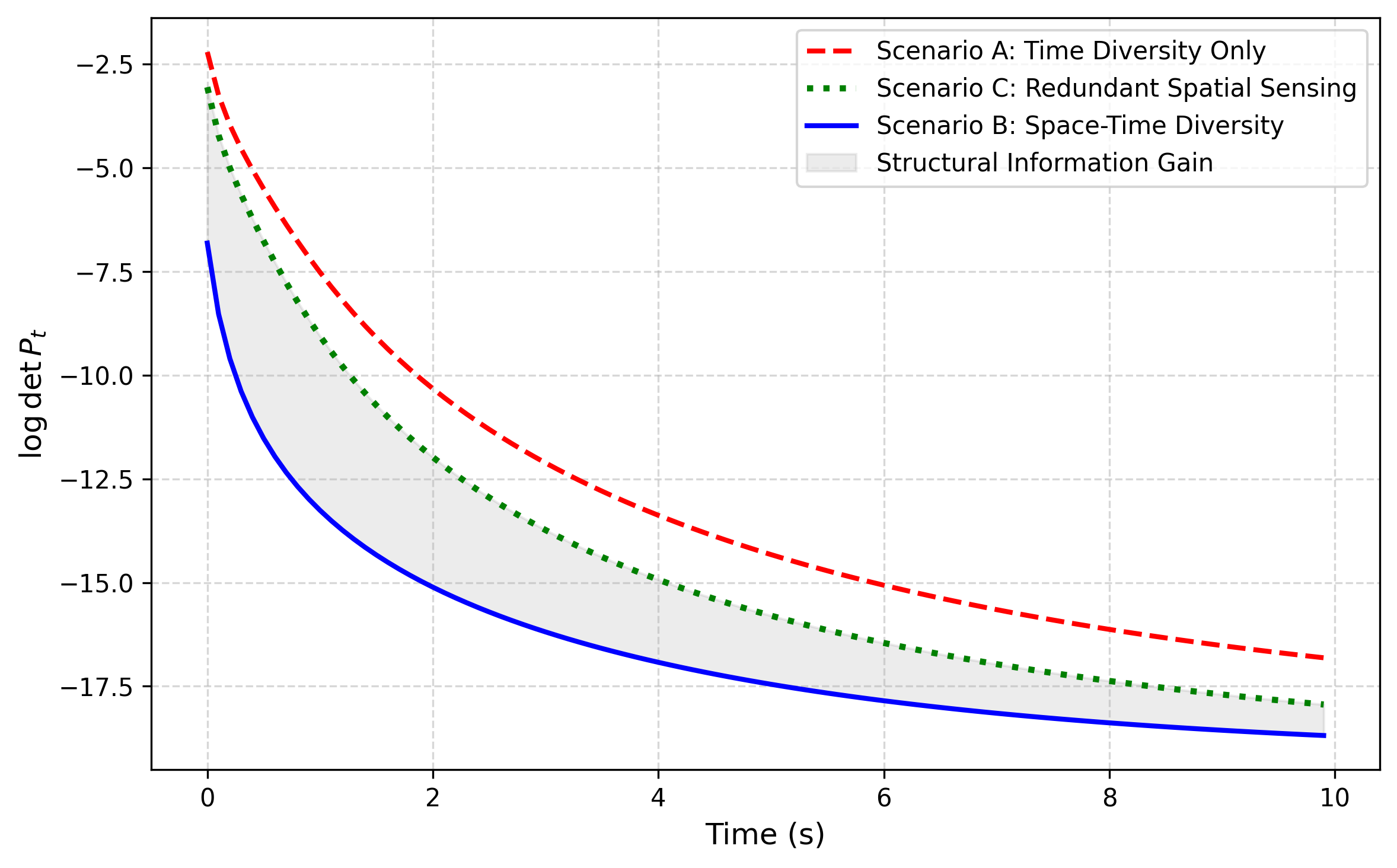} 
    \caption{Evolution of $\log \det P_t$ for a system evolving on $SE(2)$. Scenario A uses time diversity only, while Scenario B incorporates space-time diversity. Scenario C incorporates redundant spatial sensing. The shaded region highlights the structural information gain from non-redundant sensors.}
    \label{fig:sim}
\end{figure}

Figure \ref{fig:sim} shows the evolution of $\log \det P_t$ (simulated with $dt=0.1$s, process noise covariance $Q_t = \mathrm{diag}(10^{-6}, 10^{-4}, 10^{-4})$, and measurement noise $R_{\text{pos}} = 0.5 I_2$, $R_{\text{head}} = 0.01$, $R_{\text{redundant}} = 0.5$). Scenario B shows a massive, immediate drop in uncertainty at $t=0$, illustrating the direct structural spatial diversity gain. The shaded gap corresponds to the cumulative spatial information predicted by the Decomposition Theorem. Conversely, as predicted by the Saturation Theorem, Scenario C experiences no initial structural drop. It tracks Scenario A closely, yielding only a marginal asymptotic improvement due to standard Euclidean noise averaging rather than a fundamental expansion of the observability subspace.

In addition to $\log \det P_t$, we observe that the smallest eigenvalue of $P_t$ increases significantly faster in Scenario B, indicating improved conditioning of the estimation problem due to non-redundant spatial excitation.

This behavior directly reflects Theorem 1 and Theorem 2: Scenario B introduces new directions outside the propagated subspace, while Scenario C remains confined within it.

These results have direct implications for sensor architecture design. In coupled dynamical systems, sensor placement should not be determined solely by coverage of individual state components, but by how measurements interact with system dynamics. A sensor that excites a previously unobserved Lie algebra factor can yield immediate structural gains, while additional sensors aligned with existing observation subspaces provide only marginal benefit. This suggests that optimal sensor configurations should be co-designed with system dynamics to maximize information propagation across state components.

\section{Extension to 3D Navigation on $SE(3)$}
To demonstrate the scalability of the structural framework, we extend the analysis to 3D kinematics on the special Euclidean group $SE(3) \cong SO(3) \times \mathbb{R}^3$, the standard state space for autonomous vehicles, drones, and spacecraft. 

The state comprises a rotation matrix $R \in SO(3)$ and a translation vector $p \in \mathbb{R}^3$.
The Lie algebra is:
\begin{equation}
   \nonumber
     \mathfrak{se}(3) \cong \mathfrak{so}(3) \oplus \mathbb{R}^3 
\end{equation}
with elements parameterized by  
\begin{equation}
   \nonumber
   \xi = [\omega^\top, v^\top]^\top \in \mathbb{R}^6
\end{equation}
Under the left-invariant error formulation, the discrete-time state transition operator for the error $\delta x_t = [\delta \theta_t^\top, \delta p_t^\top]^\top$ is given by the Adjoint matrix of the inverse incremental motion:
\begin{equation}
    \Phi(t+1, t) = \text{Ad}(\exp(-u_t \Delta t)) = 
    \begin{bmatrix}
        R_t^\top & 0 \\
        -R_t^\top p_t^\wedge & R_t^\top
    \end{bmatrix},
\end{equation}
where $u_t = [\omega_t^\top, v_t^\top]^\top$ is the body-frame velocity input, $R_t$ and $p_t$ are the incremental rotation and translation over the interval $\Delta t$, and $(\cdot)^\wedge$ maps a vector to its skew-symmetric matrix.

Applying the structural results from Section III yields precise observability guarantees for 3D navigation:

\subsubsection{GPS-Only Navigation (Coupling Necessity)}
Consider a system equipped solely with a position sensor (e.g., GPS), yielding a factor-local measurement on $\mathbb{R}^3$ with $\tilde{H}_t = I_3$. By \textbf{Theorem 1}, the unmeasured orientation factor $\mathfrak{so}(3)$ becomes observable if and only if the cross-factor transition block is persistently exciting. Here, 
\begin{equation}
    \Phi_{\mathbb{R}^3 \leftarrow \mathfrak{so}(3)}(t+1, t) = -R_t^\top p_t^\wedge.
\end{equation}
This block is nonzero only when the vehicle translates ($p_t \neq 0$). Furthermore, full cross-factor observability (Corollary 1) requires the vehicle's trajectory to span at least two non-collinear directions; otherwise, rotations about the axis of translation remain unobservable in the null space of $p_t^\wedge$.

\subsubsection{Adding an Attitude Sensor (Spatial Diversity)}
If an attitude sensor (e.g., a star tracker or magnetometer) is added, it provides a direct factor-local measurement on $\mathfrak{so}(3)$. Because this measurement channel acts on a previously unmeasured Lie algebra factor, its row space is orthogonal to the GPS-propagated subspace $\mathcal{S}_K$ at $t=0$. By \textbf{Theorem 2}, this guarantees immediate spatial diversity gain, allowing the observability Gramian's rank to jump and the estimation uncertainty to drop instantaneously, rather than waiting for temporal accumulation via non-collinear translations.

\subsubsection{Redundant GPS (Diversity Saturation)}
Conversely, placing a second GPS antenna on the vehicle provides an additional observation channel on $\mathbb{R}^3$. However, its linearized operator falls precisely into the existing spatial subspace already covered by the first GPS. By the \textbf{Spatial Diversity Saturation Theorem}, this provides redundant temporal noise reduction but zero structural expansion of the observability Gramian.

Table \ref{tab:se3_rank} summarizes the structural observability outcomes for these architectures, highlighting how dynamic coupling and spatial diversity dictate the unobservable subspace dimensions on $\mathfrak{se}(3)$.

\begin{table}[htbp]
\caption{Structural Observability on $SE(3)$}
\label{tab:se3_rank}
\centering
\renewcommand{\arraystretch}{1.5} 
\begin{tabular}{p{0.3\columnwidth} p{0.3\columnwidth} p{0.3\columnwidth}}
\hline
\textbf{Sensor Architecture} & \textbf{Kinematic Condition} & \textbf{Unobservable Subspace} \\
\hline
GPS Only & Hovering \newline ($p_t = 0$) & $\mathfrak{so}(3)$ \newline (Dimension 3) \\
GPS Only & Translating \newline ($p_t \neq 0$) & Axis of translation \newline (Dimension 1) \\
GPS + Attitude & Any motion & $\emptyset$ \newline (Dimension 0) \\
Redundant GPS & Translating \newline ($p_t \neq 0$) & Axis of translation \newline (Dimension 1) \\
\hline
\end{tabular}
\end{table}

\section{Discussion and Future Work}
The proposed framework establishes a structural baseline for observability on product Lie groups, but several open avenues remain. First, the current analysis relies on deterministic, time-varying transition operators. Extending these bounds to systems with stochastic coupling---such as jump-Markov Lie group dynamics---would provide tighter guarantees for highly volatile navigational environments. Second, the spatial diversity saturation theorem currently assumes finite-dimensional groups. Generalizing these conditions to infinite-dimensional Lie groups, such as those encountered in continuum robotics or fluid-structure interactions, represents a formidable theoretical challenge.

Furthermore, these results suggest that observability is jointly determined by the sensor architecture and the system's realized dynamics. Since information propagation across Lie algebra factors occurs through the transition operators, the executed trajectory itself becomes a critical design variable. Even with a fixed sensor suite, trajectories that sufficiently excite the system can enhance cross-factor information flow, whereas poorly conditioned motions may delay or inhibit observability. This highlights a direct link between estimation and control in shaping the structure of the observable subspace. 

Consequently, integrating these necessary-and-sufficient structural conditions directly into active sensing and sensor-scheduling algorithms could yield computationally efficient, geometry-aware trajectory planners. Because evaluating the structural subspace condition $\text{Range}((H_{K+1})^\top) \subseteq \mathcal{S}_K$ relies purely on algebraic Lie bracket generation and unobservable subspace propagation, it bypasses the need for expensive, high-dimensional numerical Riccati covariance updates during the initial sensor-selection phase.

\section{Conclusion}
In summary, our results establish a rigorous structural framework for understanding the limits and potential of sensor diversity in dynamical state estimation on product Lie groups. We show that robust inference is governed not only by the number and quality of sensors, but critically by the alignment between observation subspaces and the system’s intrinsic coupling structure. This perspective exposes fundamental constraints that are invisible to classical Euclidean analysis, revealing that nominal observability is insufficient for guaranteeing effective estimation. Instead, the propagation of information across system factors—and the ability of time to compensate for missing spatial diversity—are dictated by the geometry and dynamics of the underlying group structure. These insights provide a principled foundation for designing sensor architectures and estimation algorithms that fully exploit the interplay between local sensing and global coupling, enabling more reliable and efficient inference in complex, distributed systems.

\section*{Acknowledgment}
The authors gratefully acknowledge Professor Vijay G. Subramanian (Department of Electrical Engineering and Computer Science, University of Michigan, Ann Arbor) for his valuable insights and productive discussions.


\end{document}